\documentclass[11pt,a4paper]{article}
\usepackage{amssymb,amsmath, amsfonts}
\usepackage{graphicx,graphics}
\usepackage{mathtools}
\usepackage{bbold}
\usepackage[english]{babel}
\usepackage[utf8]{inputenc}
\usepackage{epsfig,url}
\usepackage{bbm,theorem}
\usepackage{a4wide}
\usepackage{color}
\usepackage{enumerate}
\usepackage{calrsfs}
\usepackage[bookmarks=true]{hyperref}
\usepackage{bookmark}
\usepackage{amsmath}
\usepackage{amsfonts}
\usepackage{amssymb}
\usepackage{verbatim}
\usepackage{graphicx}

\providecommand{\U}[1]{\protect\rule{.1in}{.1in}}
\DeclareMathAlphabet{\pazocal}{OMS}{zplm}{m}{n}
\newtheorem{theorem}{Theorem}[section]
\newtheorem{definition}[theorem]{Definition}

\newtheorem{lemma}[theorem]{Lemma}
\newtheorem{proposition}[theorem]{Proposition}

{\theorembodyfont{\upshape}
\newtheorem{remark}[theorem]{Remark}

}
\numberwithin{equation}{section}
\numberwithin{theorem}{section}
\newcommand{\qed}{\hfill$\Box$}
\newenvironment{proof}{\begin{trivlist}\item[]{\em Proof:}\/}{\qed\end{trivlist}}

\newcommand{\R}{{\mathbb R}}

\newcommand{\N}{{\mathbb N}}

\DeclareMathOperator*{\supp}{supp}

\newcommand{\beq}{\begin{equation}}
\newcommand{\eeq}{\end{equation}}
\newcommand{\beqs}{\begin{eqnarray}}
\newcommand{\eeqs}{\end{eqnarray}}

\newcommand{\norm}[1]{\Vert #1\Vert}

\newcounter{jlisti}

\begin{document}

\title{Coagulation equations for aerosol dynamics}
\author{Marina A. Ferreira}
\maketitle

\begin{abstract}
Binary coagulation is an important process in aerosol dynamics by which two particles merge to form a larger one. The distribution of particle sizes over time may be described by the so-called Smoluchowski's coagulation equation. This integrodifferential equation exhibits complex non-local behaviour that strongly depends on the coagulation rate considered.
We first discuss well-posedness results for the Smoluchowski's equation for a large class of coagulation kernels  as well as the existence  and nonexistence of stationary solutions in the presence of a source of small particles. The existence result uses  Schauder fixed point theorem, and the nonexistence result relies on a flux formulation of the problem and on power law estimates for the decay of stationary solutions with a constant flux. We then consider a more general setting. We consider that particles may be constituted by different chemicals, which leads to multi-component equations describing the distribution of particle compositions. We obtain explicit solutions in the simplest case where the coagulation kernel is constant by using a generating function. Using an approximation of the solution we observe that the mass localizes along a straight line in the size space for large times and large sizes. 
\end{abstract}

\tableofcontents

\bigskip

\bigskip 

\section{Introduction}

We consider particle systems where moving particles undergo binary coagulation, forming larger particles. This simple system can be used to study the dynamics of aerosols in the atmosphere \cite{Fried} as well as raindrop formation, smoke, sprays and galaxies~\cite{Aldous1999, Elimelech1995BOOK}. 

When the number of particles is  very large it becomes more relevant to study the collective behaviour of the particles rather than individual particle behaviour. This motivates a statistical description of the system. In 1916 Smoluchowski \cite{S16} proposed an equation to describe the particle size distribution over time, assuming that the system is homogeneous in space.

 Let $f(x,t)$ be the number density of particles of size $x>0$ at time $t \geq 0$. The Smoluchowski's coagulation equation, or simply coagulation equation, is the following mean-field equation for the evolution of $f$
\begin{equation}\label{eq:smol}
\partial_t f(x,t) = \frac{1}{2}\int_{0}^x K(x-y,y)f(x-y,t)f(y,t)dy - \int_0^\infty K(x,y)f(x,t)f(y,t)dy
\end{equation} 
where $K(x,y)$ is the coagulation rate between  particles of size $x$ and $y$.  The first term on the right hand-side is the gain term due to the coagulation between particles of size $x-y$ and particles of size $y$ to create a particle of size $x$. The second term is the loss term which describes the loss of particles of size $x$ by merging with any other particle in the system. 
Equation \eqref{eq:smol} is an integrodifferential equation belonging to the class of kinetic equations. 

 We also consider more general systems where a constant input of particles may be present. 
The number density in this case satisfies the coagulation equation with an extra source term $\eta \geq 0$, 
\begin{equation}\label{eq:smol_source}
\partial_t f(x,t) = \frac{1}{2}\int_{0}^x K(x-y,y)f(x-y,t)f(y,t)dy - \int_0^\infty K(x,y)f(x,t)f(y,t)dy + \eta(x).
\end{equation}

Complementary research lines have expanded over the last decades on experimental, numerical and theoretical aspects of equations \eqref{eq:smol} and \eqref{eq:smol_source}. 
Algorithms to simulate these equations  have been developed to test hypotheses drawn from atmospheric data \cite{LK2003, Vehkam} (see \cite{Lee} for a survey on numerical methods). On the other hand, theoretical results have clarified issues mainly related to equation \eqref{eq:smol}, such as  existence and uniqueness of solutions  for general classes of kernels \cite{FL06, Norris} or the behaviour of  solutions for explicitly solvable kernels \cite{MP04} and general kernels \cite{BNV18, BNV19, EM05,FL05}.

A particle may be characterized not only by its size but also by its composition, leading to \textit{multi-component coagulation equations} where the size is described by a vector $x \in \R_+^d\backslash \{ 0\}$  representing the size of each of the chemical components of a particle. An application of multi-component equations to aerosol dynamics is described in Section \ref{sec:aerosol}. 

In this paper we review analytic results related to the  one-component equations \eqref{eq:smol} and \eqref{eq:smol_source} as well as to the corresponding  discrete  multi-component equations with $x, y \in \N^d\backslash \{ 0\}$. 
We start in Section \ref{sec:prelimi} with a short overview on various topics related to properties of the solutions, applications and derivation from particle systems. We also introduce some notation that is used throughout the Chapter.
In Section \ref{sec:onecomponent} we study the one-component equations \eqref{eq:smol} and \eqref{eq:smol_source}. 
Section \ref{sec:well-posedness} contains the main steps of the proof of one of the first well-posedness results for equation \eqref{eq:smol} with unbounded coagulation kernels obtained in 1999 by Norris \cite{Norris}. Section \ref{sec:stationary} contains a review of the proofs of existence and non-existence of stationary solutions to coagulation equations to \eqref{eq:smol_source}, obtained recently in \cite{FLNV19}.
In Section \ref{sec:multi} we consider the discrete multi-component equation with constant kernel. Following the computations presented in \cite{KBN},  we compute in Section \ref{sec:multi_time_dep} explicit time-dependent solutions 
and in Section \ref{sec:multi_stat} we compute stationary solutions when an additional source at the monomers is present. We also obtain approximations of both solutions showing explicitly that mass localizes along a straight line in the multi-dimensional size space for large times and large sizes. Finally in Section \ref{sec:perspectives} we mention some recent results in the literature and open questions.

\section{Preliminaries}\label{sec:prelimi}

\subsection{Conservation of mass and continuity equation}

By multiplying \eqref{eq:smol} by $x$ and integrating in $x$ from $0$ to $\infty$ one obtains formally an equation for the mass $M_1(t) = \int_0^\infty xf(x,t)dx$ given by $\frac{d}{dt} M_1(t) = 0 $. This shows that the mass is conserved, provided  the integrals are well-defined. Associated to the mass-conservation, one may write a continuity equation that shows that  mass is transported continuously along the size space: 
\begin{equation}\label{eq:continuity_eq}
\partial_t( xf(x,t)) = \partial_x J(x,t)
\end{equation}
where the flux of mass  from small to large clusters is given by
\begin{equation}
J(x,t) = \int_0^x \int_{x-y}^\infty yf(y,t) f(z,t)K(y,z) dzdy.
\end{equation}
As we will see in Section \ref{sec:exist} a (non-equilibrium) stationary solution has a constant flux of mass at large sizes, i.e., $J(x)$ is constant for all  $x>L$, for some positive $L$.  Moreover this flux plays an important role in the proof of non-existence of stationary solutions in Section \ref{sec:non-exist}.

Interestingly, if the coagulation rate  behaves like a power law and if the power  is sufficiently large, then  mass-conservation is lost.
Such phenomenon is  called gelation and it corresponds to the formation of infinitely large clusters that are not seen any more by the equation. Therefore these big clusters leave the system and the total mass decreases.
 Gelation may be interpreted as a change in state from gas to gel.   Mathematically, this phenomenon poses interesting challenges \cite{EMP02}. Since gelation has not been observed in atmospheric aerosols we do not discuss it further here.

We note that, contrarily to the Boltzmann equation, the coagulation equation does not preserve number of particles, due to the sticky collisions.

\subsection{Coagulation kernels for aerosols in the atmosphere}\label{sec:aerosol}

Atmospheric aerosols are suspensions of small particles in the air, whose diameter ranges approximately between $1$ nanometre, in the case of molecular particles, to $100$ micrometres, in the case of cloud droplets and dust particles \cite{Fried}.  
Aerosols influence sunlight scattering by reflecting and absorbing radiation, and they constitute the seeds that originate the clouds. Therefore, they play an important role in weather and climate forecast \cite{Carslaw2013}.

Aerosols are subject to complex processes that influence their size distribution over time. One important process is the coagulation of particles to produce larger ones. Other processes include the formation of new small particles, or monomers, due to certain physical and chemical processes, the removal of  particles due to gravity or diffusion, and the growth/shrinkage due to condensation/evaporation \cite{LK2003}. 
Atmospheric aerosols are typically constituted by different chemicals, leading to multi-component systems, which may alter the rate of the processes mentioned before and consequently, the particle size distribution \cite{Vehkam}.  

We consider the regime in which the particles are uniformly distributed in space. Moreover, we assume that removal and  growth of particles due to condensation is not important, which in practice may correspond roughly to sizes between $10$ nanometers and $10 $ micrometers \cite{Fried}. 
We are then led to the study of multi-component systems where particles undergo binary coagulation in the presence of a source of small particles. 

Coagulation kernels $K$ have been derived for atmospheric aerosols using kinetic theory under several assumptions on the shape and motion of  particles \cite{Fried}. Aerosol particles are commonly assumed to be spherical and to undergo elastic collisions with background air particles. The number of such collision events is assumed to be  much larger than the number of collisions between two coalescing particles. This drives the system towards an equilibrium where the particle velocities follow a Maxwell-Boltzmann distribution. 

 Moreover, any collision between coalescing particles yields a coalescing particle.
 Two different coagulation kernels have been derived under the previous conditions for two different  regimes. Each regime is defined based on the relation between particle size and the average distance travelled by a particle between two collisions in air, called mean free path.
 Under normal pressure and temperature conditions, the mean free path in  air, $\ell$, is of the order of $ 10$ nanometres. 
If the size (diameter) of a spherical particle, $x$, is much smaller than the mean free path $x \ll \ell$, the particle is more  likely to travel in straight lines before meeting another coalescing particle. In this case the rate of coagulation has been estimated by the free molecular coagulation kernel:
\begin{equation}\label{eq:free_mol_kernel}
K(x,y) = (x^{1/3}+y^{1/3})^2(x^{-1}+y^{-1})^{1/2}.
\end{equation}
Otherwise, if the size of a particle is much larger than the mean free path, $x \gg \ell$,  the coalescing particle will meet many background air particles before meeting another coalescing particle. In this case,  the air behaves like a fluid and the coalescing particle is more likely to diffuse. The coagulation rate has been estimated by the diffusive coagulation kernel:
\begin{equation}\label{eq:diff_kernel}
K(x,y) = (x^{-1/3}+y^{-1/3})(x^{1/3}+y^{1/3}).
\end{equation}
This kernel was first derived in the original work by Smoluchowski \cite{S16}. 
Other kernels have been derived under different assumptions on the underlying background gas and particles, such as particles moving in a laminar shear or turbulent flow \cite{Fried}, and particles having electric charges \cite{SB73, Vehkam}.

The behaviour and even the existence of solutions to  equation \eqref{eq:smol_source} strongly depends  on the coagulation kernel. 
In Sections \ref{sec:exist} and \ref{sec:non-exist} we review the existence of stationary solutions for a large class of kernels which includes in particular the free molecular \eqref{eq:free_mol_kernel} and the diffusive kernels \eqref{eq:diff_kernel}.

\subsection{From particle models to  Smoluchowski's coagulation equation }

The Smoluchowski's coagulation equation has been rigorously derived using  different approaches that consider different types of particle systems. In one approach,  a purely stochastic particle system is considered, where pairs of particles are randomly picked to originate a new particle. The associated stochastic process is usually called Marcus-Lushnikov process. A different approach considers deterministic particle systems, where particles move and when they collide they merge with a certain probability.

The first approach is inspired in Kac-models for the derivation of the Boltzmann equation \cite{FG04}. A common strategy is to start from an infinite stochastic particle system where particles of size $x$ and $y$ coalesce at a rate $K(x,y)$  and to prove that the number density, after being conveniently rescaled,  converges, as the unit volume tends to infinity, to a measure that solves the Smoluchowski's coagulation equation with kernel $K$. This has been obtained for the additive kernel, product kernel as well as for a class of sub-multiplicative kernels using combinatorial techniques and random graphs. See \cite{B06} (Chapter 5.2) for an accessible exposition and \cite{Aldous1999} for a review on existing results and open problems.

In the second approach, there are fewer rigorous results. The first result to the best of our knowledge is due to Lang and Xanh \cite{LX80}. They consider  Brownian particles moving in the three-dimensional Euclidean space according to Brownian motion with a diffusion coefficient $D$. The particles are assumed to move independently on each other provided they are at a distance greater than the sum of their radius $2R$. Once they come closer than $2R$ they coalesce with probability $1/2$, forming one Brownian particle with the same radius $R$ and the same diffusion coefficient $D$. In the limit when the number of particles $N$ goes to infinity and the radius $R$ goes to zero, such that $RN$ remains constant, the authors prove propagation of chaos and that the density function converges in probability to the solution to the Smoluchowski's coagulation equation with constant coagulation kernel. The limit where $RN$ remains constant is the so-called Boltzmann-Grad limit and  is the limit of constant mean free time.
 A more general case of coalescing Brownian particles with  diffusion coefficients changing after coalescence, but not the size $R$, has been treated in \cite{HR07}. More recently, the change in size after coalescence has been considered in \cite{NV17} (see also \cite{NNTV})  in the case of a  tracer particle moving in a straight line and coalescing with randomly distributed fixed particles of different sizes. In this case, a linear coagulation equation with a simple shear kernel was derived in the kinetic limit where the volume fraction filled by the background of particles tends to zero.

\subsection{Notation}

We  use the notation $\R_* := (0,\infty)$, $\R_+ := [0,\infty)$ and $\N_+ = \{ 0,1,...\}$. We denote the $\ell_1- $ norm in $\R^d$ by $|\alpha| = \sum_{i=1,...,d} |\alpha_i|,\ \alpha \in \R^d$.
We denote by $\mathcal{M}(I)$ the space of signed Radon measures supported on $I \subset \R_+$, i.e., the non-negative measures having finite total variation on any compact subset of $I$, and by  $\| \cdot \|$ the total variation norm. We denote by $\mathcal{M}_+(I)$ the space of measures on $\mathcal{M}(I)$ that are nonnegative. The  measures from $\mathcal{M}_+(I)$ that are also bounded are denoted by $\mathcal{M}_{+,b}(I):=\{\mu\in \mathcal{M}_{+}(I) \,|\,\mu(I)<\infty\}$. The space $\mathcal{M}_{+,b}(I)$ equipped with the norm $\| \cdot \|$ is a Banach space. The notation $f_t(x)$ will sometimes be used to denote $f(t,x)$. We denote by $C_{c}(I )$ or $C_b(I)$ the spaces of  continuous functions  on $I$ that are compactly supported  or  bounded, respectively.
For simplicity,  we use a generic constant $C>0$ which may change from line to line.

\bigskip

\section{One-component equation}
\label{sec:onecomponent}

\subsection{Main results}

We consider kernels $K:(0,\infty)^2 \to [0,\infty)$ satisfying
\begin{eqnarray}
K(x,y )=K(y,x),\  K(x,y)\geq 0,\label{eq:K_sym}\\
K(x,y) \geq c_1 [x^{\lambda +\gamma} y^{-\lambda}+ y^{\lambda +\gamma} x^{-\lambda}],\label{eq:K_growth1}\\
K(x,y) \leq c_2 [x^{\lambda +\gamma} y^{-\lambda}+ y^{\lambda +\gamma} x^{-\lambda}],
\label{eq:K_growth2}\\
0<c_1 \leq c_2 < \infty, \label{eq:K_c}\\
\text{ and } \lambda,\ \gamma \in \R,\label{eq:K_lam_gam}
\end{eqnarray}
for some given constants $c_1, c_2, \lambda$ and $\gamma$ and for all $(x,y)\in (0,\infty)^2$.
This class includes in particular the physical kernels   \eqref{eq:free_mol_kernel} and \eqref{eq:diff_kernel}.
The parameter $\gamma$ represents the homogeneity of the kernel, while $\lambda$ represents the "off-diagonal" rate. The parameter $\gamma$ yields the behaviour under the scaling of the particle size, while $\lambda$ measures the importance of collisions between particles of different sizes. Note that the bounds in \eqref{eq:K_growth1} and \eqref{eq:K_growth2} are homogeneous, i.e., they satisfy for any $k>0$, $h(kx,ky) = k^\gamma h(x,y)$, but  the kernels are not necessarily  homogeneous. 

We assume the following condition on the source  $\eta \in \mathcal{M}_+(\R_*)$ 
\begin{equation}\label{eq:cond_eta}
\supp \eta \in [1,L],\ \text{ for some } L>1. 
\end{equation}
Note that then the source is bounded, i.e., $\eta(\R_*) < \infty$.

 We consider the following definition of time-dependent solution to the Smoluchowski's coagulation equation \eqref{eq:smol} \cite{Norris}.

\begin{definition}\label{def:time-dep}
Assume that $K$ is a measurable function satisfying \eqref{eq:K_sym} and \eqref{eq:K_growth2}. 
We will say that the map $t \mapsto f_t;\ [0,T) \to \mathcal{M}_+(\R_*)$, where $T \in (0,\infty]$ is a local solution to \eqref{eq:smol} if it satisfies
\begin{enumerate}
\item for all compact sets $B \subset \R_*$, the map $t \mapsto f_t(B);\ [0,T) \to [0,\infty)$ is measurable
\item for all $t <T$ and all compact sets $B \subset \R_*$ 
$$\int_0^t \int_{B \times \R_*} K(x,y) f_s(dx)f_s(dy) ds<\infty,$$ 
\item  for all bounded measurable functions $\varphi$ of compact support  and   $t<T$ it holds
\begin{equation}\label{eq:timeDep}
\langle \varphi,f_t\rangle = \langle \varphi, f_0 \rangle + \int_0^t \langle \varphi, L(f_s) \rangle ds
\end{equation}
where $L(f)$ is defined by 
$$\langle \varphi,L(f)\rangle= \frac{1}{2} \int_{\R_*} \int_{ \R_*} K(x,y) [ \varphi(x+y) - \varphi(x) - \varphi(y)] f(dx) f(dy),$$
\item $\int_{\R_*} x \mathbbm{1}_{x \leq 1} f_0(dx) < \infty$ and \eqref{eq:timeDep} holds with $\varphi(x) = x \mathbbm{1}_{x \leq 1}$. 
\end{enumerate}
If $T = \infty$ we call time-dependent solution to \eqref{eq:smol}. 
\end{definition}

One can easily check that  condition $2$ is the minimal one to have well-defined integrals.
 Condition $3$ is the weak formulation commonly used in the literature and it is obtained by formally multiplying \eqref{eq:smol} by a test function and integrating in $x$. Condition $4$ is a boundary condition imposing that no mass enters at $0$. 

The existence and uniqueness  of a time-dependent solution to \eqref{eq:smol} for sublinear kernels is established in the next Theorem \cite{Norris}. Similar results have  also been proven in \cite{Norris} for kernels satisfying \eqref{eq:K_growth1}-\eqref{eq:K_growth2} with $\gamma+\lambda = -\lambda $ and $\lambda >- 1/2$.

\begin{theorem}\label{thm:well-posed}
Let $K$ be a measurable function satisfying \eqref{eq:K_sym} and \eqref{eq:K_growth2} with  $\lambda = 0$ and $\gamma < 1$. 
If $\langle x^2,f_0 \rangle < \infty$, then  there exists a unique time-dependent solution $(f_t)_{t>0}$ to \eqref{eq:smol} in the sense of definition \ref{def:time-dep}. 
\end{theorem}

We consider now a source $\eta \neq 0$ of small particles entering into the system at a constant rate as described by equation \eqref{eq:smol_source}. We study the existence of stationary injection solutions, i.e., solutions that satisfy  $f(t,x)=f(0,x)$ for all $t>0$, as defined next \cite{FLNV19}.

\begin{definition}
\label{DefFluxSol} Assume that $K:{\R}_*^{2}\rightarrow {\mathbb{R}%
}_{+}$ is a continuous function satisfying \eqref{eq:K_sym} and the upper bound %
\eqref{eq:K_growth2}. Assume further that $\eta \in \mathcal{M}_{+}\left( \mathbb{%
R}_*\right) $ satisfies \eqref{eq:cond_eta}. We will say that $f\in 
\mathcal{M}_{+}\left( \mathbb{R}_*\right) ,$ satisfying $f\left( \left(
0,1\right) \right) =0$ and
\begin{equation}
\int_{\R_* }x^{\gamma +\lambda }f\left( dx\right) + \int_{\R_* }x^{-\lambda }f\left( dx\right) <\infty\,,
\label{eq:moment_cond}
\end{equation}
is a stationary injection solution of \eqref{eq:smol_source} if the following
identity holds for any test function 
$\varphi \in C_{c}({\R}_*)$: 
\begin{equation}
\frac{1}{2}\int_{\R_*}\int_{\R_*} K\left( x,y\right) \left[
\varphi \left( x+y\right) -\varphi \left( x\right) -\varphi \left( y\right) %
\right] f\left( dx\right) f\left( dy\right) +\int_{\R_*}\varphi
\left( x\right) \eta \left( dx\right) =0\,.  \label{eq:stationary_eq}
\end{equation}
\end{definition}

Condition \eqref{eq:moment_cond} is the minimal one for the integrals in \eqref{eq:stationary_eq} to be well-defined. 
Stationary injection solutions have a constant in time flux of mass  from small to large sizes, due to the source, therefore they are non-equilibrium solutions. Note that to be able to be stationary, the volume of particles entering the system has to balance the volume of particles leaving the system. Interestingly, there is an implicit removal of particles from the system at infinite sizes that allows the existence of these solutions. 
As we will see in the next two Theorems, for some class of coagulation rates, including the diffusive kernel \eqref{eq:diff_kernel}, such balance exists, while for other class of kernels, including the free molecular kernel \eqref{eq:free_mol_kernel}, such balance does not exist.

\begin{theorem}
\label{thm:existence} Assume that $K$ satisfies~\eqref{eq:K_sym}--\eqref{eq:K_lam_gam} and $| \gamma +2\lambda | <1.$ Let $\eta \neq 0  $ satisfy \eqref{eq:cond_eta}. Then, there exists a stationary injection solution $f\in 
\mathcal{M}_{+}\left( \mathbb{R}_*\right) $, $f\neq 0$, to~%
\eqref{eq:smol_source} in the sense of Definition~\ref{DefFluxSol}. 
\end{theorem}

\begin{theorem}
\label{thm:NonExistence} Suppose that $K $ satisfies~\eqref{eq:K_sym}--\eqref{eq:K_lam_gam} as well as $| \gamma +2\lambda | \geq 1.$ Let
us assume also that $\eta \neq 0$  satisfies \eqref{eq:cond_eta}. Then, there is not any solution of~\eqref{eq:smol_source} in the sense of  Definition \ref{DefFluxSol}.
\end{theorem}

Note that the diffusive kernel \eqref{eq:diff_kernel} satisfies the growth conditions \eqref{eq:K_growth1}-\eqref{eq:K_growth2} with $\gamma = 0$ and $\lambda = 1/3 $, while the free molecular kernel \eqref{eq:free_mol_kernel} satisfies the growth conditions with $\gamma = 1/6 $ and $\lambda = 1/2$. Therefore there exists a stationary solution for the diffusive but not for the free molecular kernel.

The mass flux from small to large sizes associated to a stationary injection solutions is given in the next Lemma. 

\begin{lemma}\label{lem:flux}
Suppose that the assumptions of Theorem~\ref{thm:existence} hold. Let $f$ be a stationary injection solution in the sense of Definition~\ref{DefFluxSol}. Then $f$ satisfies for any $R>0$ 
\begin{equation}\label{eq:flux_lem}
J(R) = \int_{(0,R]} x \eta(dx) \, 
\end{equation}
where $J(R)$ is the mass flux at size $R$ and it is defined by
\begin{equation*}
J(R) := \int_{ (0,R]}\int_{(R-x,\infty)} K(x,y)x f(dx) f(dy)
\end{equation*}
\end{lemma}
\begin{remark}
If $R\geq L_\eta$, the right-hand side of \eqref{eq:flux_lem} is constant equal to $J_\eta=\int_{[1,L_\eta]} x \eta(dx)>0$.  Therefore,  $J(R)=J_\eta$ for $R>L_\eta$, i.e., the mass flux is constant in the regions that include large sizes.
\end{remark}

The main ideas to prove Lemma \ref{lem:flux} are the following. For  
  each $\varepsilon>0$, we define the test function 
$\varphi(x) = x \chi_\varepsilon(x) \in C_{c}({\R}_*)$ where
 $\chi_\varepsilon \in C_c^\infty(\R_*)$ is such that $0\le \chi_\varepsilon\le 1$, 
$\chi_\varepsilon (x) = 1$, for $1\le x \leq R$, and $\chi_\varepsilon(x) = 0$, for $x \geq R+\varepsilon $. 
Using this test function in \eqref{eq:stationary_eq}, the result can be obtained after letting $\varepsilon \to 0$.

\subsection{Well-posedness for the time-dependent problem}
\label{sec:well-posedness}

We describe the main ideas of the proof of Theorem \ref{thm:well-posed} obtained in \cite{Norris} (Section 2).

The first step is to prove well-posedness for a truncated problem. The second step is to  obtain estimates that allow us to remove the truncation and to obtain well-posedness for the original problem. 

Let $B \subset \R_*$ be a compact set. 
Note that all measures in $\mathcal{M}(B)$ are bounded. 
Note that from the hypotheses of Theorem \ref{thm:well-posed} on the kernel we have that  
\begin{equation}\label{eq:K_cond_w}
0\leq K(x,y) \leq w(x) + w(y), \quad \text{ with } w(x) := x^\gamma \text{ and } \gamma < 1.
\end{equation} 
The truncated operator  $L^B: \mathcal{M}(B) \times \R \to \mathcal{M}(B) \times \R $ is defined by
\begin{eqnarray}
\langle (\varphi, a), L^B(f,\xi) \rangle := \hspace{11cm} \nonumber \\
 \frac{1}{2} \int_{\R_*\times \R_*} \{ \varphi(x+y) \mathbbm{1}_{ \{ x+y \in B\}} + a w(x+y) \mathbb{1}_{ \{ x+y \notin  B \}} - \varphi(x) - \varphi(y) \} \times K(x,y)f(dx)f(dy) \nonumber \\
+ \xi \int_{\R_*} \{ a w(x) - \varphi(x) \} w(x) f(dx)\nonumber
\end{eqnarray}
for all bounded measurable functions $\varphi$ on $\R_*$ and all $a \in \R$,  where  $\langle (\varphi, a), (f, \xi) \rangle $ denotes $\langle \varphi, f \rangle + a \xi $. The truncated equation reads
\begin{equation}\label{eq:trunc}
\langle (\varphi, a),(f_t,\xi_t)\rangle = \langle (\varphi, a),(f_0,\xi_0)\rangle + \int_0^t \langle (\varphi, a), L^B(f,\xi) \rangle ds.
%\int_{\R_*} \varphi(x) f_t(dx) + a\xi_t = \int_{\R_*} \varphi(x) f_0(dx) + a\xi_0 + \int_0^t  L^B[\varphi,a](f_s,\xi_s)ds.
\end{equation}

 An interpretation of the dynamics associated with operator $L^B$ is the following (see \cite{Norris} for more details). 
Particles of size $x$ and $y$ merge at a rate $K(x,y)$ and they produce a new particle of size $x+y$. If the merging particle has size outside $B$, we add $w(x+y)$.

A solution to \eqref{eq:trunc} is defined next.
\begin{definition}
Let $T \in (0,\infty )$. We will say that ${(f_t, \xi_t)}_{t \in [0, T]}$ is a local solution to \eqref{eq:trunc} if $t \mapsto (f_t,\xi_t); [0,T] \to  \mathcal{M}(B) \times \R$ is a continuous map satisfying \eqref{eq:trunc} for all $t \in [0,T]$. Additionally, $(f_t)_{t \in [0, T]}$ is called a solution  to \eqref{eq:trunc} when $[0,T]$ is replaced by $[0,\infty)$.
\end{definition}

\begin{proposition}\label{prop:Norris}
Suppose that $f_0 \in \mathcal{M}(B)$ with $f_0 \geq 0$ and $\xi_0 \in [0,\infty)$. The equation \eqref{eq:trunc} has a unique solution ${(f_t, \xi_t)}_{t \geq 0}$ starting from $(f_0, \xi_0)$. Moreover $f_t \geq 0$ and $ \xi_t \geq 0 $ for all $t\geq 0$.
\end{proposition}

We will discuss the main ideas of the proof that is organized in three steps.
The first step is to show that there is a constant $T>0$ depending only on $\gamma$ and $B$ such that there exists a unique local solution ${(f_t, \xi_t)}_{t \in [0, T]}$ to \eqref{eq:trunc} starting from $(f_0, \xi_0)$.  This is obtained by using an iterative scheme of continuous maps $(f^n_t, \xi^n_t): [0,\infty) \mapsto \mathcal{M}(B) \times \R $ defined by 
\begin{eqnarray*}
(f^0_t, \xi^0_t)&=&(f_0, \xi_0)\\
(f^n_t, \xi^n_t)&=&(f_0, \xi_0) + \int_0^t L^B(f^{n-1}_t, \xi^{n-1}_t)
\end{eqnarray*}
and proving that there exists a $T>0$ such that $(f^n, \xi^n)$ converges in $\mathcal{M}(B) \times \R$ uniformly in $t\leq T$ to the desired local solution, which is also unique.

The second step is to prove that $f_t\geq 0,\ t \in [0, T]$, which is obtained using again an iterative argument similar to the one used in the first step. 

Finally, the third step is to show that the solution exists for all times $t \in [0,\infty)$. Choosing $\varphi=w $ and $a=1$ we obtain that
\begin{equation*}
\frac{d}{dt} (\langle w , f_t \rangle + \xi_t) = \frac{1}{2} \int_{\R_*\times \R_*} \{ w(x+y) - w(x) - w(y) \}K(x,y) f_t(dx) f_t(dy) \leq 0, 
\end{equation*}
which implies that
$$\| f_T\| + |\xi_T| \leq \langle w, f_T\rangle + \xi_T \leq   \langle w, f_0\rangle + \xi_0.$$
Using a scaling argument, we may assume without loss of generality that $\langle w, f_0\rangle + \xi_0  \leq 1$, consequently $\| f_T\| + |\xi_T| \leq 1$. We can start again from $(f_T,\xi_T)$ at time $T$ to extend the solution to $[0,2T]$ and so on.
Moreover, choosing $\varphi=0$ and $a=1$ in \eqref{eq:trunc}, we obtain 
$$ \frac{d }{dt} \xi_t = \frac{1}{2} \int_{\R_*\times \R_*} \{ w(x+y) \mathbb{1}_{ \bar B}(x+y)  \} K(x,y)f(dx)f(dy) + \xi_t \int_{\R_*}  w^2(x) f(dx),
$$
which implies that $\xi_t\geq 0$, for all $t \geq 0$, due to  $f_t\geq 0$, which ends the proof of the proposition.

\bigskip

\begin{proof}[Theorem \ref{thm:well-posed}]
Fix $f_0 \in \mathcal{M}_+$, such that $\langle w, f_0 \rangle < \infty$. For each compact set $B \subset {\R_*}$ define $f_0^B = \mathbbm{1}_B f_0 $ and $\xi_0^B = \int_{\bar B}w(x) f_0(dx)$. 
From Proposition \ref{prop:Norris} there is a unique solution $(f_0^B,\xi_t^B)_{t\geq 0 }$ to \eqref{eq:trunc} starting from $(f_0^B,\xi_0^B)$. We now set $f_t = \lim\limits_{B \to {\R_*}} f_t^B$ and $\xi_t= \lim\limits_{B \to {\R_*}} \xi_t^B$.  
 Using \eqref{eq:K_cond_w}, we obtain by dominated convergence, 
$$\frac{d}{dt} \langle \varphi,f_t \rangle = \frac{1}{2} \int_{{\R_*} \times {\R_*}} \{ \varphi(x+y) - \varphi(x) - \varphi(y)\} K(x,y) f_t(dx)f_t(dy) - \xi_t \langle \varphi w, f_t \rangle,
$$ 
for all bounded measurable functions $\varphi$.
One can prove that for all $t < T$ and for any local solution $(g_t)_{t<T}$, 
\begin{equation}\label{eq:nu_ineq}
f_t \leq g_t, \quad \langle w, f_t \rangle + \xi_t \geq \langle w, g_t \rangle.
\end{equation}
By hypothesis $\langle w^2, f_0 \rangle < \infty$. Since $K(x,y) \leq w(x) + w(y)$ it holds $\langle  w^2,f_t \rangle \leq \langle  w^2,f_0 \rangle \exp(2 \langle  w,f_0 \rangle t )$.
Therefore $\langle w^2 , f_t \rangle < \infty$ for all $t>0$, which allows to pass to the limit as $B \to {\R_*}$ in \eqref{eq:trunc} and to deduce that 
\begin{equation}\label{eq:lambda0}
\xi_t = 0,\ \ t>0.
\end{equation}
Then \eqref{eq:nu_ineq} and \eqref{eq:lambda0} imply that $(f_t)_{t\geq 0}$ is a time-dependent solution to \eqref{eq:smol} and moreover, it is the only solution.
\end{proof}

\subsection{Stationary solutions with injection}\label{sec:stationary}

We now consider  the Smoluchowski's coagulation equation with source  \eqref{eq:smol_source}.

\subsubsection{Existence of stationary solutions}
\label{sec:exist}

We present here the main ideas of the proof of the existence Theorem \ref{thm:existence} obtained in \cite{FLNV19}.
The general strategy  is similar to the strategy used in the proof of well-posedness presented in the previous Section. First, we prove existence of a stationary solution for a truncated problem and second, we obtain estimates that allow to remove the truncation and hence the existence result for the original problem. Unfortunately the method used to prove existence does not give uniqueness, that problem needs a separate treatment (see \cite{L15} for a simple explanation of the available techniques).

Let $\varepsilon>0$ and $R_* \geq L_\eta $, where $L_\eta$ is the upper bound of the support of the source $\eta$ defined in \eqref{eq:cond_eta}. We will eventually make $\varepsilon \to 0$ and $R_* \to \infty$. We consider kernels $K_{\varepsilon,R_*}$ that are continuous, bounded and have compact support, such that
\begin{eqnarray}
K_{\varepsilon,R_*}(x,y) \leq a_2(\varepsilon), &\quad (x,y) \in \R^2_+ \label{eq:K_eps_R1}\\
K_{\varepsilon,R_*}(x,y) \in [a_1(\varepsilon),a_2(\varepsilon)], &\quad (x,y) \in [1,2R_*] \label{eq:K_eps_R2} \\
K_{\varepsilon,R_*}(x,y) =0, &\quad x \geq 4R_* \text { or } y  \geq 4R_*, \label{eq:K_eps_R3}
\end{eqnarray} 
and
\begin{equation}
\lim\limits_{R_* \to \infty} K_{\varepsilon,R_*}(x,y) = K_{\varepsilon}(x,y) \label{eq:K_eps_R4}
\end{equation}
where $K_{\varepsilon} $ is continuous and satisfies $K_{\varepsilon}(x,y) \in [a_1(\varepsilon),a_2(\varepsilon)],$ for all $(x,y) \in \R_+^2$ and
\begin{equation}
\lim\limits_{\varepsilon \to 0} K_{\varepsilon}(x,y) = K(x,y).\label{eq:K_eps}
\end{equation}

Additionally, in the evolution equation, we consider a cut-off of the gain term due to the coagulation that ensures that the measure solutions are supported in $[1,2R_*]$ and bounded at all times. 
To this end, we choose $\zeta _{R_{\ast }}\in C\left(\mathbb{R}_*\right)$ such that $0\le \zeta _{R_{\ast }}\le 1$, $\zeta_{R_{\ast }}\left( x\right) =1$ for $0\leq x\leq R_{\ast }$, and $\zeta _{R_{\ast}}\left( x\right) =0$ for $x\geq 2R_{\ast }$.
The regularized time evolution equation  then reads as 
\begin{equation}
\partial _{t}f(x,t)=\frac{\zeta _{R_{\ast }\!}(x) }{2}\int_{\left(
0,x\right] }K_{\varepsilon,R_*}( x-y,y) f( x-y,t) f( y,t)
dy - \int_{{{\R}_*}}\! K_{\varepsilon,R_*}( x,y) f( x,t) f(
y,t) dy+\eta( x) \,. \label{evolEqTrunc}
\end{equation}

\begin{definition}\label{DefTimeSol}  Let $\varepsilon>0$ and $R_* \geq L_\eta $. Suppose that $K_{\varepsilon,R_*}$ satisfies \eqref{eq:K_eps_R1}-\eqref{eq:K_eps_R3} and $\eta \in \mathcal{M}_{+}({\mathbb{R}}_{+})$ satisfies~\eqref{eq:cond_eta}.
Consider some initial data $f_{0}\in \mathcal{M}_{+}({\mathbb{R}}_*)$ for which $f_{0}\left( \left( 0,1\right) \cup \left( 2R_{\ast
},\infty \right) \right) =0$. Then $f_0\in \mathcal{M}_{+,b}({\mathbb{R}}_*)$.
We will say that $f\in {C^{1}(\left[ 0,T%
\right] ,\mathcal{M}_{+,b}({\mathbb{R}}_*))}$ satisfying $f\left( \cdot
,0\right) =f_{0}\left( \cdot \right) $
is a time-dependent solution of \eqref{evolEqTrunc} if the following
identity holds for any test function $\varphi \in C^{1}(\left[ 0,T\right]
,C_{c}\left( {\mathbb{R}}_*\right) )$ and all $0<t<T$,
\begin{eqnarray}
&&\frac{d}{dt}\int_{{{\mathbb{R}}_*}}\varphi \left( x,t\right) f\left(
dx,t\right) -\int_{{{\mathbb{R}}_*}}\dot{\varphi}\left( x,t\right) f\left(
dx,t\right)  \notag \\
&=&\frac{1}{2}\int_{{{\mathbb{R}}_*}}\int_{{{\mathbb{R}}_*}}K_{\varepsilon,R_*}\left(
x,y\right) \left[ \varphi \left( x+y,t\right) \zeta _{R_{\ast }}\left(
x+y\right) -\varphi \left( x,t\right) -\varphi \left( y,t\right) \right]
f\left( dx,t\right) f\left( dy,t\right) \hspace{1cm}  \notag \\
&&+\int_{{{\mathbb{R}}_*}}\varphi \left( x,t\right) \eta \left( dx\right),
\label{eq:evol_eqWeak}
\end{eqnarray}
where $\dot \varphi$ denotes the Fréchet time-derivative of $\varphi$.
\end{definition}

\begin{proposition}
\label{thm:existence_evolution} 
Let $\varepsilon>0$ and $R_* \geq L_\eta $. Suppose that $K_{\varepsilon,R_*}$ satisfies \eqref{eq:K_eps_R1}-\eqref{eq:K_eps_R3} and $\eta \in \mathcal{M}_{+}({\mathbb{R}}_{+})$ satisfies~\eqref{eq:cond_eta}. Then, for any initial
condition $f_{0}$ satisfying $f_{0}\in \mathcal{M}_{+}({\mathbb{R}}_*)$, $%
f_{0}\left( \left( 0,1\right) \cup \left( 2R_{\ast },\infty \right) \right)
=0$ there exists a unique time-dependent solution 
$f\in {C^{1}(\left[ 0,T\right] ,\mathcal{M}_{+,b}({\mathbb{R}}_*))}$
to \eqref{evolEqTrunc} which solves it in the classical
sense. Moreover, $f$ is a Weak solution of \eqref{evolEqTrunc} in the sense of Definition  \ref{DefTimeSol}
such that
\begin{equation*}
f\left( \left( 0,1\right) \cup \left( 2R_{\ast },\infty \right) ,t\right)
=0\, ,\quad \text{for }0\leq t \le T\,,% \label{fVanish}
\end{equation*}%
and the following estimate holds
\begin{equation*}
\int_{{\mathbb{R}}_*}f(dx,t)\leq  \int_{{\mathbb{R}}_*}f_0(dx)+ C t \, ,\quad t\geq 0 \,, \label{fEstim}
\end{equation*}%
for $C = \int_{{\mathbb{R}}_*}\eta(dx)\ge 0$ which is independent of $f_{0}$, $t$, and $T$.
\end{proposition}

To prove Prop. \ref{thm:existence_evolution} we observe that  since the kernel is bounded, the result may be obtained using Banach fixed-point theorem.

\begin{definition}
\label{DefStatSol} Let $\varepsilon>0$ and $R_* \geq L_\eta $. Suppose that $K_{\varepsilon,R_*}$ satisfies \eqref{eq:K_eps_R1}-\eqref{eq:K_eps_R3} and $\eta \in \mathcal{M}_{+}({\mathbb{R}}_{+})$ satisfies~\eqref{eq:cond_eta}.
We will
say that $f\in {\mathcal{M}_{+}({\mathbb{R}}_*)},$ satisfying $f((0,1) \cup (2R_*,\infty))=0$
is a stationary injection solution of \eqref{evolEqTrunc} if the following
identity holds for any test function $\varphi \in C_{c}\left( {\mathbb{R}}%
_*\right) $: 
\begin{eqnarray*}
0 &=&\frac{1}{2}\int_{{\mathbb{R}}_*}\int_{{\mathbb{R}}_*}K_{\varepsilon,R_*}\left(
x,y\right) \left[ \varphi \left( x+y\right) \zeta _{R_{\ast }}\left(
x+y\right) -\varphi \left( x\right) -\varphi \left( y\right) \right] f\left(
dx\right) f\left( dy\right)  \notag \\
&&+\int_{{\mathbb{R}}_*}\varphi \left( x\right) \eta \left( dx\right) .
\label{eq:evol_eqWeakSt}
\end{eqnarray*}
\end{definition}

We  denote by $S(t)$ the semigroup defined by the time-dependent solution $f$ obtained in Proposition \ref{thm:existence_evolution}, 
$$S(t) f_0 = f(\cdot, t)$$
that satisfies  the semigroup property
\begin{equation*}
S(t+s)f = S(t)S(s)f, \ t,s \in \R_+.
\end{equation*}
The operators $S(t)$ define mappings 
$$S(t): \mathcal{X}_{R_{\ast }} \to \mathcal{X}_{R_{\ast }}, \quad \text{ for each } t_1, t_2 \in \R_+$$
with $\mathcal{X}_{R_{\ast }}=\left\{ f\in \mathcal{M}_{+}({\mathbb{R}}_*):f\left( \left( 0,1\right) \cup \left( 2R_{\ast },\infty \right)\right) =0\right\}$.

\begin{proposition}
\label{thm:existence_truncated} Under the assumptions of Proposition~\ref%
{thm:existence_evolution}, there exists a stationary injection solution $\hat{f}%
\in \mathcal{M}_{+}({\mathbb{R}}_*)$ to \eqref{evolEqTrunc} as defined in
Definition~\ref{DefStatSol}.
\end{proposition}

\begin{proof} [Idea of the proof]
The key point of the proof is to use Schauder fixed point theorem.
The first step is to obtain the existence of an invariant region for the evolution problem \eqref{eq:evol_eqWeak}. To that end, we choose a time independent test function $\varphi(x) = 1$ for $x \in [1, 2R_*]$. Using the lower bound for the kernel \eqref{eq:K_eps_R2} and that $f(\cdot , t)$ has support in $[1, 2R_*]$ we obtain the following estimate
\begin{equation*}
\frac{d}{dt}\int_{[1,2R_{\ast }]}f\left( dx,t\right) \leq -\frac{a_{1}}{2}%
\left( \int_{[1,2R_{\ast }]}f\left( dx,t\right) \right) ^{2}+c_{0}
\end{equation*}%
where $c_{0}=\int_{\R_*}\eta \left( dx\right) $.
This implies that for a large enough $M>0$, the set
\begin{equation*}
\mathcal{U}_{M}=\left\{ f\in \mathcal{X}_{R_{\ast }}:\int_{[1,2R_{\ast
}]}f(dx)\leq M\right\} \, . \label{eq:invariant_region}
\end{equation*}
is invariant under the time evolution \eqref{evolEqTrunc}.
Moreover, $\mathcal{U}_{M}$ is compact in the $\ast -$weak topology due to Banach-Alaoglu's Theorem (cf.\cite{Brezis}), since it is an intersection of a $\ast -$weak closed set $\mathcal{X}_{R_{\ast }}$ and the closed ball $\norm{f}\le M$.

The second step is to prove that for each $t>0$,  both maps $S(t):\mathcal{U}_{M} \to \mathcal{U}_{M}$ and $t \mapsto S(t) f_0$  are continuous in the $*-$weak topology. 

Finally, the third step of the proof reads as follows. Since for each $t$, the operator $S(t)$ is continuous and $\mathcal{U}_{M}$ is compact and convex when endowed with the $*-$weak topology, we can apply Schauder fixed point theorem to conclude that for all $\delta>0$ there is a fixed point $f_\delta$ of $S(\delta)$ in $\mathcal{U}_{M}$. Moreover, since $\mathcal{U}_{M}$ is metrizable and hence sequentially compact, there is a convergent sequence $\{ f_{\delta_n}\}_{n\in \N}$, i.e., there exists $\hat f \in \mathcal{U}_{M}$  such that $ f_{\delta_n} \to \hat f $ when $ \delta_n \to 0$ in the $*-$weak topology. 
For each $t$ we choose $\delta_n = t/n$. Using the semigroup property  we obtain that 
$S(t)f_{\delta_n} =S(n\delta_n)f_{\delta_n}=S(\delta_n)f_{\delta_n}$.
Using the continuity of $t \mapsto S(t) f_0$ and the fact that $S(0) \hat f= \hat f$, we obtain $S(\delta_n)f_{\delta_n} \to \hat f$.
On the other hand using the continuity of $S(t)$ we obtain that $S(t)f_{\delta_n} \to S(t) \hat f$. 
Therefore $S(t)\hat f = \hat f$ and thus $\hat f$ is a stationary solution to \eqref{evolEqTrunc}, which concludes the proof.
\end{proof}

The next Lemma provides uniform estimates for integrals. 

\begin{lemma}\label{lem:bound}
Let $a>0$, $R\ge a$ and $b \in (0,1)$ be such that $bR >a$. Suppose $f \in \mathcal{M_+}(\R_*)$, $\varphi \in C(\R_*)$,
$g \in L^1(\R_*)$, and $g,\varphi\ge 0$. If
\begin{equation*}\label{eq:avg}
\frac{1}{z}\int_{[bz,z]} \varphi(x) f(dx) \leq g(z)\,, \quad \text{for } z \in [a,R] \,,
\end{equation*}
then
\begin{equation*}\label{eq:bound}
\int_{[a,R]} \varphi(x) f(dx) \leq \frac{\int_{[a,\infty)}g(z)dz}{\ln (b^{-1})}
+ R g(R)\,.
\end{equation*}
\end{lemma}

We now extend the previous existence result to general unbounded kernels $K$ supported in $\R_+^2$ and satisfying the conditions of the theorem \ref{thm:existence}.

\begin{proof}[Idea of the proof of Theorem \ref{thm:existence}]
Let $f_{\varepsilon ,R_{\ast }}$ be a stationary injection solution to \eqref{evolEqTrunc} as  in
Definition~\ref{DefStatSol} provided by Proposition \ref{thm:existence_truncated}. The idea is to obtain estimates that are independent on both $\varepsilon$ and $R_{\ast }$ that allow to pass to the limit as $\varepsilon \to 0$ and $R_{\ast } \to \infty$ and to obtain the existence of a stationary injection solution to the original problem as defined in \eqref{DefFluxSol}.

First we obtain an estimate uniform in $R_*$:
\begin{equation*}
\int_{[0,2R_*/3]}f_{\varepsilon ,R_{\ast }}(dx)\leq \bar C_{\varepsilon },\ \ \ R_*> 0 ,
\label{EstTruncFunction}
\end{equation*}
where $\bar C_{\varepsilon }$ is a constant independent on $R_{\ast }.$ 
This estimate
implies, that taking a subsequence if needed, there exists $f_{\varepsilon}\in \mathcal{M}_{+}\left( \mathbb{R}_{+}\right) $ such that $f_{\varepsilon}\left( \left[ 0,1\right) \right) =0$ and:
\begin{equation*}
f_{\varepsilon ,R_{\ast }^{n}}\rightharpoonup f_{\varepsilon }\text{ as } n\rightarrow \infty \text{ in the }\ast -\text{weak topology}
\label{fWeakLimit}
\end{equation*}
with $R_{\ast }^{n}\rightarrow \infty $ as $n\rightarrow \infty .$
 For any bounded continuous test function $\varphi: [0,\infty) \to \R$, one proves that $f_{\varepsilon }$ satisfies 
\begin{equation*}
\frac 1 2\int_{\lbrack 0,\infty )^{2}}K_{\varepsilon }\left( x,y\right) [\varphi
(x+y)-\varphi (x)-\varphi (y)]f_{\varepsilon }\left( dx\right)
f_{\varepsilon }\left( dy\right) +\int_{[0,\infty )}\varphi (x)\eta \left(
dx\right) =0.  \label{WeakFormEps}
\end{equation*}
where $K_{\varepsilon }$ is defined in \eqref{eq:K_eps_R4}.

Second, we obtain  estimates independent on $\varepsilon$:
\begin{equation}
\frac{1}{z}\int_{\left[  \frac{2z}{3},z\right]  }f_{\varepsilon}\left(
dx\right)  \leq\frac{\tilde{C}}{z^{\frac{3}{2}}}\left(  \frac{1}{\min\left(
z^{\gamma},\frac{1}{\varepsilon}\right)  }\right)  ^{\frac{1}{2}}, \label{A1} 
\end{equation}
and
\begin{equation*}
\frac{1}{z}\int_{\left[  \frac{2z}{3},z\right]  }f_{\varepsilon}\left(
dx\right)  \leq\frac{\tilde{C}}{z^{\frac{3}{2}}\sqrt{\varepsilon}} \label{A2}
\end{equation*}
where $\tilde{C}$ is independent on $\varepsilon$.
This estimate yields $\ast -$weak compactness of the family of measures $\left\{
f_{\varepsilon }\right\} _{\varepsilon >0}$ in $\mathcal{M}_{+}\left( 
\mathbb{R}_{+}\right) .$ Therefore, there exists $f\in \mathcal{M}_{+}\left( 
\mathbb{R}_{+}\right) $ such that:%
\begin{equation*}
f_{\varepsilon _{n}}\rightharpoonup f\text{ as }n\rightarrow \infty \text{
in the }\ast -\text{weak topology}  \label{fepsWeakLimit}
\end{equation*}%
for some subsequence $\left\{ \varepsilon _{n}\right\} _{n\in \mathbb{N}}$
with $\lim_{n\rightarrow \infty }\varepsilon _{n}=0.$
Using \eqref{eq:K_eps} and Lemma~\ref{lem:bound}, one can prove that $f_{\varepsilon }$ satisfies \eqref{eq:stationary_eq} for any $\varphi \in C_{c}\left( \mathbb{R}_{+}\right)$. In particular, $f\neq 0$ due to $\eta \neq 0$.

It only remains to prove (\ref{eq:moment_cond}).
Taking the limit of (\ref{A1}) as $\varepsilon \rightarrow 0$ we arrive at:
\begin{equation*}
\frac{1}{z}\int_{[2z/3,z]}f(dx)\leq \frac{\widetilde C}{z^{3/2+\gamma /2}}\ \ \text{ for
all }z\in (0,\infty ),
\end{equation*}
which implies 
\begin{equation*}
\frac{1}{z}\int_{[2z/3,z]}x^{\mu} f(dx)\leq \widetilde C \frac{z^{\mu} }{z^{3/2+\gamma /2}}\ \ \text{ for
all }z\in (0,\infty ),
\end{equation*}
for any $\mu \in \R$.
From Lemma~\ref{lem:bound} we obtain the boundedness of the  moment of order $\mu$:
\begin{equation*}\label{eq:moment_mu}
\int_{[0,\infty)}x^{\mu} f(dx)< \infty.
\end{equation*}
 for any $\mu$ satisfying $\mu < \frac{\gamma + 1}{2}$. In particular, since $|\gamma + 2\lambda| <1$, then the moments  $\mu = -\lambda$ and $\mu = \gamma + \lambda$ are bounded, which proves (\ref{eq:moment_cond}).
\end{proof}

\subsubsection{Nonexistence of stationary solutions}
\label{sec:non-exist}

We present the main ideas of the proof of Theorem  \ref{thm:NonExistence} obtained in \cite{FLNV19}.
The proof is done by contradiction. Let the kernel $K$ satisfy the power law bounds \eqref{eq:K_growth1}-\eqref{eq:K_growth2} with $|\gamma+2\lambda| \geq 1 $. Suppose that  $f\in \mathcal{M}_{+}\left( \mathbb{R}_*\right)$  is a stationary injection solution of \eqref{eq:smol} in the sense of Definition \ref{DefFluxSol}. Then $f$ satisfies the weak formulation \eqref{eq:stationary_eq} as well as the condition on the moments \eqref{eq:moment_cond}.

The first step is to rewrite \eqref{eq:stationary_eq} using the flux formulation.   
Consider the function $J:\mathbb{R}_{*}\rightarrow \mathbb{R}_{+}$ defined by
\begin{equation}
J\left( R\right) =\iint_{\Sigma _{R}}K\left( x,y\right) xf\left(dx\right) f\left( dy\right)  \label{S4E5a}
\end{equation}
where 
\begin{equation*}
\Sigma _{R}=\left\{ x\geq 1,\ y\geq 1:x+y>R,\ x\leq R\right\}.
\end{equation*}
Let $\varepsilon > 0,\ R \geq 1 $ and $\chi_\varepsilon \in C^\infty(\R_+)$ satisfy $\chi_\varepsilon (x) = 1,\ x \leq R$ and $\chi_\varepsilon(x) = 0,\ x \geq R+\varepsilon $. Choosing a test function $\varphi(x) = x \chi_\varepsilon(x) $ we obtain from~\eqref{eq:stationary_eq}
 the flux formulation
\begin{equation*}
J(R) = \int_{[1,R]} x \eta (dx), \ R\geq 1. 
\end{equation*}
We note that $J$ describes the flux of particles passing through $\Sigma _{R}$ and that this flux is constant for all $R\geq L_\eta$ and equal to $J\left( L_\eta\right) = \int_{[1,\infty)} x \eta (dx)>0$, i.e., 
\begin{equation*}
J\left( R\right) = J(L_\eta),\ R\geq L_\eta.
\end{equation*} 
The second step is to prove that the main contribution to the integral \eqref{S4E5a} as $R \to \infty$ is due to collisions between particles of size close to $R$ and particles of size of order $1$.
To that end, for a given $\delta>0$
small, we consider a partition of $\Sigma_R = D_{\delta}^{\left( 1\right) } \cap D_{\delta}^{\left( 2\right) }$ such that
\begin{align*}
D_{\delta}^{\left( 1\right) } & =\left\{ x\geq1,\ y\geq1:y\leq\delta x\right\} \ , \\
D_{\delta}^{\left( 2\right) } & =\left\{ x\geq1,\ y\geq1:y>\delta x\right\} \ 
\end{align*}
and we define 
\begin{equation*}
J_k(R)  = \iint_{\Sigma _{R} \cap D_{\delta}^k}K\left( x,y\right) xf\left(dx\right) f\left( dy\right),\ k=1,2.
\end{equation*}
Therefore
\begin{equation*}
J(R) = J_1(R) + J_2(R). 
\end{equation*}
Using the upper bound for the kernel \eqref{eq:K_growth2}, the moment condition \eqref{eq:moment_cond} and the fact that $\Sigma _{R} \cap D_{\delta}^2 \subset [1,R]\times [\frac{\delta R}{1+\delta },\infty )$ one concludes after some computations that the contribution of $J_2$ vanishes as $R\to \infty$, i.e.,
\begin{equation*}
\lim\limits_{R \to \infty} J_2(R) = 0
\end{equation*}
which implies that 
\begin{equation*}
\lim\limits_{R \to \infty} J_1(R) = \lim\limits_{R \to \infty} J(R) = J(L_\eta).
\end{equation*}

In the remainder of the proof we will use the notation $a:= \gamma+\lambda$ and $b:= -\lambda$ if $(\gamma + 2\lambda) \geq 1$, or $a:= \gamma+\lambda$ and $b:= -\lambda$ if $(\gamma + 2\lambda) \leq -1$. Then, the assumption \eqref{eq:moment_cond} may be rewritten as
\begin{equation}
\int_{\R_* }x^af\left( dx\right)  <\infty.
\label{eq:moment_cond_2}
\end{equation}
The third step of the proof consists in obtaining a lower bound for the fluxes that implies a lower bound for the number of particles in some region of the size space. Using the upper bound for the kernel \eqref{eq:K_growth2} we obtain after some computations
\begin{equation}\label{eq:liminf}
\liminf_{R\rightarrow \infty }\left( R^{a +1}\iint_{\Sigma
_{R}\cap D_{\delta }^{\left( 1\right) }} y^{b }f\left( dx\right)
f\left( dy\right) \right) \geq \frac{J\left( L_\eta\right) }{c_3\left( 1+\delta
^{|a-b| }\right) } \ .
\end{equation}
For $R$ sufficiently large we have that
\begin{equation*}
\Sigma_{R}\cap D_{\delta}^{\left( 1\right) }\subset\left\{ \left( x,y\right)
:1\leq y\leq\delta R,\ R<x+y,\ x\leq R\right\}
\end{equation*}
whence, \eqref{eq:liminf} implies the inequality
\begin{equation}
\int_{\left[ 1,\delta R\right] }y^{b}f\left( dy\right) \int_{\left(
R-y,R\right] }f\left( dx\right) \geq\frac{J\left( L_\eta\right) }{2c_3\left(
1+\delta^{|a-b|}\right) }\frac{1}{R^{a+1}}
\label{S4E7}
\end{equation}
 for $R\geq R_{0}$ with $R_{0}$ large enough.
We now consider two cases separately $a \geq 0 $ and $a<0$. 
Let first $a \geq 0 $. 
Due to \eqref{eq:moment_cond_2} we may define 
\begin{equation}
F\left( R\right) =\int_{\left( R,\infty\right) }f\left( dx\right) \ \ ,\ \
R\geq1 \ . \label{S4E7a}
\end{equation}
Using \eqref{S4E7a} we can rewrite \eqref{S4E7} as
\begin{equation*}
-\int_{\left[ 1,\delta R\right] }\left[ F\left( R-y\right) -F\left( R\right) %
\right] y^{b}f\left( dy\right) \leq -\frac{J\left( L_\eta\right) }{%
2c_3\left( 1+\delta ^{|a-b| }\right) }\frac{1}{R^{a +1}}\ \ \text{for }R\geq R_{0}.
\end{equation*}
Then, using a comparison argument (see Lemma 4.1 in \cite{FLNV19}), for some constant $B>0$, it follows that
\begin{equation}
F\left( R\right) \geq \frac{B}{R^{a}}\ \ \text{if \ }
R\geq R_{0},\ \text{ for } a>0,  \label{S5E3_a}
\end{equation}
and
\begin{equation}
F\left( R\right) \geq B \log(R)\ \ \text{if \ }
R\geq R_{0},\ \text{ for } a=0.  \label{S5E3_b}
\end{equation}
In the case $a>0$, \eqref{S5E3_a} implies
$$B \leq R^a F(R) \leq \int_{(R,\infty)} x^a f(dx)$$
Taking the limit when $R \to \infty$ and using \eqref{eq:moment_cond_2}  it follows that $B\leq 0$, which leads to a contradiction.
In the case $a=0$, the contradiction follows from \eqref{S5E3_b} in a similar way using \eqref{eq:moment_cond_2}. 

Let now $a <0$. We define the function $F$ by
\begin{equation}
F\left( R\right) =\int_{[1,R] }f\left( dx\right) \ \ ,\ \
R\geq1.  \label{S4E7aa}
\end{equation}
Using (\ref{S4E7aa}) we can rewrite (\ref{S4E7}) as: 
\begin{equation}
-\int_{\left[ 1,\delta R\right] }\left[ F\left( R\right) -F\left( R-y\right) %
\right] y^{b }f\left( dy\right) \leq -\frac{J\left( L_\eta\right) }{%
2c_3\left( 1+\delta ^{|a-b| }\right) }\frac{1}{R^{a +1}}\ \ \text{for }R\geq R_{0}.\nonumber
\end{equation}
As in the previous case, using a comparison argument (see Lemma 4.2 in \cite{FLNV19}), it follows that 
there is $B>0$ such that 
\begin{equation*}
F\left( R\right) \geq \frac{B}{R^{a}}\ \ \text{if \ }
R\geq  R_{0} \, .  \label{S5E3a}
\end{equation*}
For a small $\varepsilon>0$ satisfying $\varepsilon < B$ there exists $M$ such that
\begin{equation*}
\int_{[M,\infty)} x^a f(dx) = \varepsilon \ .\label{S5E8a}
\end{equation*} 
Then for all $R>M$ we have 
$$B \leq  R^a \int_{[1,R]} f(dx)  
\leq  R^a \int_{[1,M]} f(dx) + \int_{[M,R]} x^{a} f(dx)
\leq   R^a \int_{[1,M]} f(dx) + \varepsilon.
$$
Since $a<0$, taking the limit as $R \to \infty$ we obtain $B \leq \varepsilon$, which leads to a contradiction.

\bigskip

\section{Discrete multi-component coagulation equation with constant kernel}
\label{sec:multi}

In this Section we consider discrete coagulation equations where the particle size is a discrete variable  representing the number of monomers.  
In addition, we consider that  particles may be composed of different types of monomers.
A particle with $d$ components is  described by a vector $\alpha=(\alpha_i)_{i=1,...,d} \in \N_+^d\backslash \{ 0 \}$ where $\alpha_i$ represents the size of the $i$th component.
 The number density $n_\alpha(t)$ of particles with composition $\alpha$ at time $t \geq 0$ satisfies the {\it discrete multi-component coagulation equation}
\begin{equation}\label{eq:multi_smol}
\partial_t n_\alpha(t) = \frac{1}{2}\sum_{0<\beta<\alpha} K_{\alpha-\beta,\beta}n_{\alpha-\beta}(t) n_\alpha(t) - \sum_{\beta>0} K_{\alpha,\beta}n_\alpha(t) n_\beta(t)
\end{equation} 
where $K_{\alpha,\beta}$ represents the coagulation rate between  particles with compositions $\alpha$ and $\beta$. We use $\alpha < \beta$ to denote $\alpha_i \leq \beta_i,\ i=1,...,d,$ and $\alpha \neq \beta$. 

The continuous counterpart of the discrete multi-component coagulation equation may be written as  \eqref{eq:multi_smol}  by substituting the sums by multidimensional integrals and by taking the discrete vector $\alpha \in \N_+^d\backslash \{ 0 \}$  as a continuous variable $\alpha \in \R_+^d \backslash \{0\}$.

In the following Sections, we consider discrete multi-component equations with the constant kernel $K_{\alpha,\beta} = 2$. We obtain an explicit time-dependent solution to \eqref{eq:multi_smol} and an explicit stationary solution in the presence of a source, by following the computations presented in \cite{KBN}. Similar results could also be obtained in the case of continuous multi-component equations.

\subsection{Mass localization in time-dependent solution}
\label{sec:multi_time_dep}

We consider the multi-component discrete equation  with  constant kernel $K_{\alpha,\beta} = 2$,
\begin{equation}\label{eq:discrete}
\frac{d n_\alpha(t)}{dt}(t) = \sum_{\beta <\alpha} n_\beta(t) n_{\alpha-\beta}(t) - 2 n_\alpha(t) \sum_{\beta>0} n_\beta(t)
\end{equation}
and initial condition
\begin{equation}\label{eq:disc_ini}
n_\alpha(0) = \frac{1}{d} \sum_{|\beta |=1} \delta_{\alpha,\beta}
\end{equation}
where $\delta_{\alpha,\beta} = 1$ if $\alpha=\beta$ and  $\delta_{\alpha,\beta} = 0$ otherwise. We recall that $|\beta| = \sum_{i=1}^d \beta_i$.

\begin{remark}
Note that the initial condition \eqref{eq:disc_ini} is supported at monomers and its mass is  uniformly distributed by the types of particles. The initial mass of each type of particle is   $\frac{1}{d}$.
\end{remark}

Existence and uniqueness of a solution to \eqref{eq:discrete}-\eqref{eq:disc_ini} in the one-component case $d=1$ is proven in \cite{MP04}  for any initial condition satisfying $\sum_{\alpha = 1}^\infty n_\alpha(0)<\infty$ using Laplace transforms. An explicit solution to the multi-component problem has been obtained in \cite{KBN} and \cite{L} using a generating function.

In this Section we review the computations described in \cite{KBN} to obtain an explicit solution to \eqref{eq:disc_sta}. We then study the long-time behaviour using an approximation of the solution for large times and large sizes as in \cite{KBN}. In particular, we observe the phenomenon of mass localization along a straight line in the size space.

Multiplying \eqref{eq:disc_sta} formally by a test function $\psi_\alpha $  and summing in $\alpha$ we obtain the weak formulation
\begin{equation}\label{eq:weak_dic}
\partial_t \sum_{\alpha=1}^\infty \psi_\alpha n_\alpha(t) =  \sum_{\alpha,\beta=1}^\infty [\psi_{\alpha+\beta} - \psi_\alpha - \psi_\beta ] n_\alpha(t)n_\beta(t).
\end{equation}
The solution  may be obtained using the generating function defined next. 
 
The generating function  $F:D\times \R_+ \to \R$ associated to a sequence $\{ n_\alpha(t)\}_{\alpha\in \N^d \backslash \{ 0\} }$ is defined by
\begin{equation}\label{eq:F}
F(z,t) = \sum_{\alpha>0} z^\alpha n_\alpha(t)
\end{equation}
where $D \subset \R^d$ is the domain of convergence of the series and  $z^\alpha= z_1^{\alpha_1} z_2^{\alpha_2}...z_d^{\alpha_d}$.

Using $\psi_\alpha = z^\alpha$ in \eqref{eq:weak_dic} we obtain an equation for the generating function $F$,
\begin{equation}\label{eq:F_eq}
\partial_t F(z,t) = F(z,t)^2 - 2F(z,t)N(t)
\end{equation}
where $N(t) = F(0,t) = \sum_{\alpha>0} n_\alpha(t)$ is the total number of particles at time $t$. From \eqref{eq:disc_ini} the initial number of particles is $N(0)=1$. Using $\psi_\alpha = 1$ in \eqref{eq:weak_dic} we obtain an equation for $N$,
\begin{equation}\label{eq:N}
\partial_t N(t) = -N^2(t),\ N(0)=1 \quad \iff \quad N(t)= \frac{1}{1+t}.
\end{equation}
If we subtract equations \eqref{eq:F_eq} and \eqref{eq:N} we obtain an equation for $F-N$. More precisely, we get $\partial_t (F-N) = (F-N)^2$. Solving this equation and using \eqref{eq:N} yields an expression for $F$ 
\begin{equation}\label{eq:F_expression}
F(z,t) = \frac{F_0(z)}{(1+t)(1+t-tF_0(z))}
\end{equation}
where $F_0(z) = F(z,0)$ is given by $F_0(z) = \frac{1}{d} \sum_{i=1}^d z_i$ after substituting \eqref{eq:disc_ini} in \eqref{eq:F}. The expression for $F$ will  be used in the following to determine the solution to \eqref{eq:discrete}. 

We note that if $\{ n_{\alpha } \}_{\alpha>0}$  is a solution to the multi-component coagulation equation \eqref{eq:discrete}, then $\{ n_{|\alpha |} \}_{\alpha>0}$, where $|\alpha| =  \sum_{i=1}^d \alpha_i$ is the sum variable and $n_{|\alpha|}$ is defined by $n_{|\alpha|} = \sum_{\beta>0} n_\beta \delta_{|\alpha|,|\beta|}$,  is a solution to the one-component equation with constant kernel $K_{|\alpha|,|\beta|}=2$ and initial condition $n_{|\alpha|}(0) = \delta_{|\alpha|,1}$. This result may be obtained using the weak formulation \eqref{eq:weak_dic} with a test function of the form $\psi_\alpha = \varphi_{|\alpha|}$. 
We first solve the one-component equation to find an expression for $\{ n_{|\alpha |} \}_{\alpha>0}$.

We consider the generating function $f: D\times \R_+ \to \R$, $D \subset \R$, associated to the one-component problem 
\begin{equation}\label{eq:F_oneComp}
f(z,t) = \sum_{|\alpha|=1}^\infty z^{|\alpha|}n_{|\alpha|}(t),
\end{equation}
which may be expressed by \eqref{eq:F_expression} with $f_0(z,t) = \sum_{|\alpha|=1}^\infty z^{|\alpha|}n_{|\alpha|}(0) = z$, i.e. 
\begin{equation}
f(z,t) = \frac{z}{(1+t)(1+t-tz)}
\end{equation}
Using the Taylor series, we expand $f$ around $z=0$ and obtain
\begin{equation}\label{eq:F_series}
f(z,t) = \sum_{k=1}^\infty z^k\frac{t^{k-1}}{(1+t)^{k+1}}, \quad z < \frac{1+t}{t},\ t>0.
\end{equation}%{\rd I am confusing, is $z$ one component or several? }
Comparing each term of the two series \eqref{eq:F_series} and \eqref{eq:F_oneComp} we conclude that the solution to the one-component equation is 
\begin{equation}\label{eq:sol_1d}
n_{|\alpha|}(t) = \frac{t^{|\alpha|-1}}{(1+t)^{|\alpha|+1}}.
\end{equation}

The solution to the multi-component equation \eqref{eq:discrete} can now be computed by expanding \eqref{eq:F_expression} and comparing with \eqref{eq:F}.
Using the Taylor series in several variables we obtain the expansion of \eqref{eq:F_expression} around $0$,
\begin{equation}
f(z,t) = \sum_{k=1}^\infty \frac{1}{d^k}\frac{t^{k-1}}{(1+t)^{k+1}} |z|^k, \quad |z| < d \frac{1+t}{t},\ t>0.
\end{equation}
Comparing with \eqref{eq:F_oneComp} and using \eqref{eq:sol_1d} and the fact that $(z_1+...+z_d)^k = \sum\limits_{|\alpha|=k} \frac{k!}{\alpha_1!\alpha_2!...\alpha_d!} z_1^{\alpha_1} z_2^{\alpha_2}...z_d^{\alpha_d}$ we finally obtain the solution to the multi-component coagulation equation \eqref{eq:discrete} expressed in terms of $n_{|\alpha|}$,
\begin{equation}\label{eq:sol_multiD}
n_\alpha(t)  = n_{|\alpha|}(t)g(\alpha) \quad \text{ with }\quad g(\alpha) =  \frac{1}{d^{|\alpha|}}\frac{|\alpha|!}{\alpha_1!\alpha_2!...\alpha_d!}.
\end{equation}

To study the long time behaviour, we use the fact that $\lim (\frac{1+t}{t})^t = e$ to obtain an approximation for $n_{|\alpha|}(t)$ for large $|\alpha|$ and large time $t$ 
\begin{equation}\label{eq:approx_n1d}
n_{|\alpha|}(t) \approx t^{-2} \exp(-\frac{|\alpha|}{t}). 
\end{equation}
\begin{remark}
In \cite{MP04} it  is shown that the solution to the continuous one-component equation with constant kernel does approach the form $f(x,t)= t^{-2} \exp\left(-\frac{x}{t}\right)$ for large times,  provided the initial mass is either finite, which includes the case treated in this Section, or its mass distribution function diverges sufficiently weakly.
\end{remark}
We also consider an approximation of the function $g$
\begin{equation}\label{eq:approx_g}
g(\alpha) \approx |\alpha|^{-(d-1)/2} \exp(-\frac{|\alpha|_-^2}{2|\alpha|})
\end{equation}
where $|\alpha|_-^2 = \frac{1}{d} \sum\limits_{i,j=1}^d (\alpha_i-\alpha_j)^2$ denotes the generalized mass difference variable.
Using \eqref{eq:approx_n1d} and \eqref{eq:approx_g} in \eqref{eq:sol_multiD} we obtain for large $t$ and $|\alpha|$ the approximation
\begin{equation}
n_\alpha(t)  \approx t^{-2} |\alpha|^{-(d-1)/2} \exp(-\frac{|\alpha|}{t})\exp(-\frac{|\alpha|_-^2}{2|\alpha|}).\label{eq:n_approx}
\end{equation}
We observe that, besides the mass scale $|\alpha | \sim t$ imported from the solution to the one-component equation,  there is a second mass scale given by $|\alpha|_- \sim \sqrt{t}$. 
Introducing the variables $\xi = \frac{|\alpha|}{t}$ and $\rho =  \frac{|\alpha|_-}{\sqrt{t}}$ we may then write the solution in a scaling form 
\begin{equation}
n_\alpha(t)  \approx t^{-(d+3)/2} \phi(\xi,\rho)
\end{equation}
where
\begin{equation}
 \phi(\xi,\rho) = \xi^{-(d-1)/2} \exp(-\xi)\exp(-\frac{\rho^2}{2\xi}).
\end{equation}
Finally we note from \eqref{eq:n_approx} that for any fixed time, $n_\alpha(t)$ reaches maximum values when $|\alpha|_-^2 = 0$. This condition defines a straight line in the size space given by $\{\alpha \in \N_+^d\ |\ \alpha_1 = \alpha_2 = ... =\alpha_d\}$, indicating that mass concentrates along a line for large times and large sizes.

\subsection{Mass localization in stationary solutions}
\label{sec:multi_stat}

Mass localization is also  observed in stationary solutions to coagulation equations with source by applying a similar analysis as in the previous Section. 
We consider the  stationary multi-component  coagulation equation with source and constant kernel  $K_{\alpha,\beta}=2$, 
\begin{equation}\label{eq:disc_sta}
0 = \sum_{\beta<\alpha } n_{\alpha-\beta}n_\beta -2\sum_{\beta>0} n_\alpha n_\beta + s_\alpha
\end{equation}
where $s_\alpha$ is the source term. In analogy to the initial conditions in the time-dependent case \eqref{eq:disc_ini}, the source term is given by
\begin{equation}
 s_\alpha = \frac{h}{d} \sum_{|\beta|=1} \delta_{\alpha,\beta},
\end{equation}
for some given $h>0$.

The constant kernel belongs to the class of kernels  considered in Section \ref{sec:onecomponent}. In particular, the constant kernel belongs to the subclass of kernels for which there is a stationary injection solution (see Theorem \ref{thm:existence}). 
Following the computations of \cite{KBN} we compute in the following an explicit solution to the multi-component equation \eqref{eq:disc_sta}.

Given a test function $\psi_\alpha$, the weak formulation is now given by 
\begin{equation}\label{eq:weak_dic_source}
0 =  \sum_{\alpha,\beta=1}^\infty [\psi_{\alpha+\beta} - \psi_\alpha - \psi_\beta ] n_\alpha n_\beta + \frac{h}{d}\sum_{|\alpha|=1} \psi_\alpha.
\end{equation}
The generating function
\begin{equation}\label{eq:F_oneComp_source}
F(z) = \sum_{|\alpha|=1}^\infty z^{|\alpha|}n_\alpha
\end{equation}
satisfies
\begin{equation}\label{eq:F_eq_source}
F(z)^2-2F(z)N+S(z)=0
\end{equation}
where $S(z) = \frac{h}{d}\sum_{|\alpha|=1} z^\alpha = \frac{h}{d} \sum_{i=1}^d z_i $ and $N=\sqrt{h}$ is obtained using an appropriate test function in \eqref{eq:weak_dic_source}. The solution to \eqref{eq:F_eq_source} reads
\begin{equation}
F(z) = \sqrt{h}[1-\sqrt{1- \frac{|z|}{d}}].
\end{equation}
The solution to \eqref{eq:disc_sta} is obtained by expanding $F$ in powers of the variables $z_i$ and comparing with \eqref{eq:F_oneComp_source}, yielding
\begin{equation}
n_\alpha= n_{|\alpha|} g(\alpha)
\end{equation}
where $g$ is defined in \eqref{eq:sol_multiD} and 
\begin{equation}
n_{|\alpha|} = \frac{\sqrt{h} (2|\alpha|)!}{(2|\alpha|-1)(2^{|\alpha|}|\alpha|!)^2}.
\end{equation}
For large sizes we may approximate $ n_{|\alpha|}$ by $\sqrt{h}|\alpha|^{-3/2}$, therefore using also \eqref{eq:approx_g}, we obtain
\begin{equation}\label{eq:n_approx_source}
n_\alpha \approx  \sqrt{h}|\alpha|^{-(d+2)/2} \exp(-\frac{|\alpha|_-^2}{2|\alpha|}).
\end{equation}
Like in the time-dependent problem, an additional size scale is observed $|\alpha|_- \sim \sqrt{|\alpha|}$.
Also here we can see from \eqref{eq:n_approx_source} that a stationary solution $n_\alpha$ reaches maximum values at the straight line defined by $\{\alpha \in \N_+^d\ |\ \alpha_1 = \alpha_2 = ... =\alpha_d\}$.
A representation of \eqref{eq:n_approx_source} is shown in Figure \ref{fig:localization}.

\begin{figure}
\center
\includegraphics[scale=0.4]{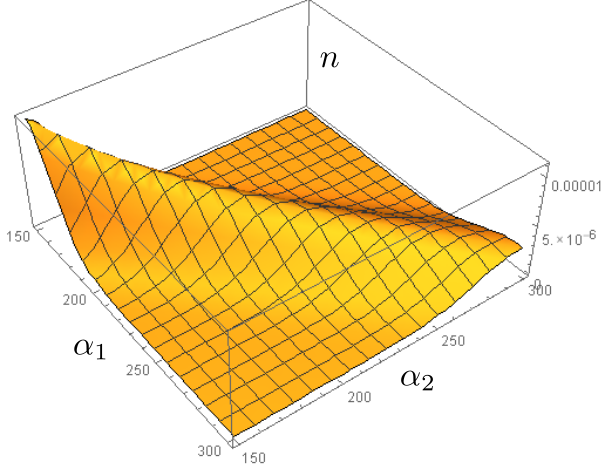}\label{fig:localization}
\caption{Approximation of a stationary solution to the two-component
coagulation equation with source and constant kernel \eqref{eq:disc_sta}. We observe a concentration of  particles  along a straight line.}
\end{figure} 

\bigskip

\section{Perspectives and open problems}\label{sec:perspectives}

The existence and uniqueness  of a time-dependent solution  have also been established in \cite{Norris} for coagulation kernels satisfying \eqref{eq:K_growth1}-\eqref{eq:K_growth2} with $\gamma+\lambda = -\lambda $ and $\lambda >- 1/2$  using a similar reasoning  as the one we presented in Section \ref{sec:exist}  \cite{Norris}.
Moreover, in \cite{EM06} existence is obtained using a functional framework, for  a class of kernels satisfying \eqref{eq:K_growth1}-\eqref{eq:K_growth2} with $c_1=c_2=1$, $\lambda \in [-1,1]$, $\gamma \in [0,2]$, $\gamma \leq -2\lambda$, $\gamma+\lambda \in [-1,1]$ and $(\gamma,\lambda)\neq (-\lambda,-1)$. In \cite{FL06} uniqueness is proved globally in time for a class of kernels satisfying some regularity conditions as well as the bounds \eqref{eq:K_growth1}-\eqref{eq:K_growth2} with $\gamma \leq 1$, $\lambda = 0$, and for a different class of kernels such that $\gamma \in (1,2]$, $\lambda = 1$ up to a gelation time $T$. 
For more general classes of kernels, both existence and uniqueness  remain  open problems.  We refer to the survey  \cite{LM} for further references.

In the  presence of a constant source of small particles, the existence and non-existence of  stationary solutions presented in Section \ref{sec:onecomponent} are the most recent existence results to the best of our knowledge. Previous results \cite{Dub} were obtained for particular classes of kernels that are included in the more general setting presented here.
%{\rd self-similar solutions}
In the case of multi-component equations with $d$ components, source and  kernel $K$ satisfying 
\begin{equation}\label{eq:K_multi}
c_1 w(x,y) \leq K(x,y) \leq c_2 w(x,y) \quad \text{ with } \quad w(x,y) = \sum_{i=1}^d x^{\gamma_i-\lambda_i}y^{\lambda_i}+y^{\gamma_i-\lambda_i}x^{\lambda_i},
\end{equation}
we expect  the existence  result (Theorem \ref{thm:existence}) to remain valid for a class of kernels satisfying $|\gamma_i+2\lambda_i| < 1$ for all $i=1,...,d$ \cite{FLNV}. In the same line, the nonexistence result (Theorem \ref{thm:NonExistence}) should hold true if $|\gamma_i+2\lambda_i| \geq 1$ for some $i$. Moreover, stationary solutions are expected to exhibit mass localization along a straight line for a  class of kernels satisfying growth bounds that are invariant under permutations of the components \cite{FLNV}. 

To the best of our knowledge, nothing is known about rigorous results for  multi-component coagulation equations with general kernels. However, the well-posedness results for the one-component case, are expected to remain true in the multi-component case provided the kernel satisfies the bounds \eqref{eq:K_multi}  
with $\gamma_i$ and $\lambda_i$ satisfying the same conditions for well-posedness in dimension $d=1$ for all $i$. The paper in preparation \cite{FLNV2} is a first step in this direction. Mass localization for large times is expected to hold for a class of kernels that satisfy the same bounds with the additional condition that $w$ is invariant under any permutation of the components.
 An additive kernel that does not satisfy this invariance was considered in \cite{YLS16}. In this case, the mass may not localize in a straight line, due to the different rates of coagulation of each component. Multiscale behaviour is then expected to emerge that could break down the nice localization structure.

Mass localization results are very important in the optimization of current algorithms as they allow to focus the computations on the region of the size space where the mass is localized. The computational complexity of the multi-component problem may in this way be reduced to the complexity of the one-component problem. 

There are also open problems on general coagulation equations with fragmentation, sink and growth terms, as well as  on the derivation of these equations from particle systems. We refer to \cite{C15} for a brief overview on some of these topics.

\bigskip 

\textbf{Acknowledgements.} 
The author is grateful to J. Lukkarinen, A. Nota and J. J. L. Velázquez for a very productive collaboration that led to the results presented in Sections \ref{sec:exist} and \ref{sec:non-exist}. 
The author acknowledges support of the Faculty of Science of University of Helsinki  through the Atmospheric Mathematics (AtMath) collaboration as well as of the Hausdorff Research Institute for Mathematics (Bonn), through the Junior Trimester Program on Kinetic Theory, of the CRC 1060 The mathematics of emergent effects at the University of Bonn funded through the German Science Foundation (DFG).

\bigskip 

\bigskip

\end{document}